\documentclass[a4paper]{cas-sc}
\usepackage{amsmath,amssymb,amsthm}
\usepackage[ruled,vlined,linesnumbered]{algorithm2e}
\usepackage[numbers,square,sort]{natbib}
\usepackage{url}
\usepackage[nobiblatex]{xurl}
\usepackage{subfigure}
\usepackage{mathtools}
\usepackage{blkarray}
\usepackage{bigstrut} 
\usepackage{multirow}
\usepackage{tikz}
\usepackage{tikz-3dplot}
\usepackage{pgfplots}
\usepackage[shortlabels]{enumitem}
\usepackage{afterpage}
\usepackage{float}

\usepackage{color}
\usepackage{todonotes}

\usepackage{algorithm2e}
\SetAlFnt{\small}
\SetAlCapFnt{\small}
\SetAlCapNameFnt{\small}
\usepackage{algorithmic}
\algsetup{linenosize=\tiny}

\pgfplotsset{compat=newest}
\usepgfplotslibrary{patchplots,statistics}
\usetikzlibrary{datavisualization}
\usetikzlibrary{datavisualization.formats.functions}

\newcommand{\CorrB}[1]{{\color{black} #1}}

{\bf}{}
{}{}
{\bf}{}
{\bf}{}
\newtheorem{prop}{Proposition}{\bf}{}
{\bf}{}
\newtheorem{cor}{Corollary}{\bf}{}
{\bf}{}
\newtheorem{defn}{Definition}{\bf}{}
\newtheorem{obs}{Observation}{\bf}{}
\newtheorem{form}{Formulation}

\def\tsc#1{\csdef{#1}{\textsc{\lowercase{#1}}\xspace}}
\tsc{WGM}
\tsc{QE}
\tsc{EP}
\tsc{PMS}
\tsc{BEC}
\tsc{DE}

\definecolor{col1}{RGB}{66,49,34}
\definecolor{col2}{RGB}{37,88,105}
\definecolor{col3}{RGB}{220,140,51}
\definecolor{col4}{RGB}{138,138,51}
\definecolor{col5}{RGB}{138,138,51}

\begin{document}
\LinesNumbered
\let\WriteBookmarks\relax
\def\floatpagepagefraction{1}
\def\textpagefraction{.001}
\shorttitle{A 2-approximation for the SPS}
\shortauthors{M. Frohn~\emph{et al.}}

\title [mode = title]{A 2-approximation algorithm for the softwired parsimony problem on binary, tree-child phylogenetic networks}

\author[1]{Martin Frohn}[orcid=0000-0002-5002-4049]
\fnmark[1]
\cormark[1]

\author[1]{Steven Kelk}[orcid=0000-0002-9518-4724]
\fnmark[2]

\address[1]{Department of Advanced Computing Sciences, Maastricht University, Paul Henri Spaaklaan 1, 6229 EN Maastricht, Netherlands.}

\cortext[cor1]{Corresponding author}
\fntext[fn1]{Email: \href{mailto:martin.frohn@maastrichtuniversity.nl}{martin.frohn@maastrichtuniversity.nl} (Martin Frohn)}
\fntext[fn2]{Email: \href{mailto:steven.kelk@maastrichtuniversity.nl}{steven.kelk@maastrichtuniversity.nl} (Steven Kelk)}

\begin{abstract}
Finding the most parsimonious tree inside a phylogenetic network with respect to a given character is an NP-hard combinatorial optimization problem that for many network topologies is essentially inapproximable. In contrast, if the network is a rooted tree, then Fitch's well-known 
algorithm calculates an optimal parsimony score for that character in polynomial time. 
Drawing inspiration from this we here introduce a new extension of Fitch's algorithm which runs in polynomial time and ensures an approximation factor of 2 on binary, tree-child phylogenetic networks, a popular topologically-restricted subclass of phylogenetic networks in the literature. Specifically, we show that Fitch's algorithm can be seen as a primal-dual algorithm, how it can be extended to binary, tree-child networks and that the approximation guarantee of this extension is tight. These results for a classic problem in phylogenetics strengthens the
link between polyhedral methods and phylogenetics and can aid in the study of other related optimization problems on phylogenetic networks.
\end{abstract}


\maketitle

\section{Introduction}\label{sec:1}
Given a set  $\Gamma =\{1,2,\dots,n\}$ of $n\geq 3$ distinct species, or more abstractly
\emph{taxa}, a \emph{phylogenetic tree} of $\Gamma$ is an ordered triplet $(T,\phi,w)$ such that $T$ is a tree having $n$ leaves, $\phi$ is a bijection between the leaves of $T$ and the taxa in $\Gamma$, and $w$ is a vector of non-negative weights associated to the edges of $T$. The high-level idea is that $T$ represents the pattern of diversification events throughout evolutionary history that give rise to the set of taxa $\Gamma$, and $w$ represents the quantity of evolutionary change along a given edge. A \emph{phylogenetic network} of $\Gamma$ is an extension of the definition of a phylogenetic tree $(T,\phi,w)$ to an ordered triplet $(N,\phi,w)$ where $N$ is a connected \CorrB{directed acyclic rooted} graph with $n$ vertices of in-degree 1 and out-degree 0. Hereinafter, given the existence of the bijection $\phi$, with a little abuse of notation, we use the terms phylogenetic tree $T$ and network $N$ (and a same symbol) to indicate both a phylogenetic tree and network and the associated graph, respectively. An example for a phylogenetic tree and network is shown in Figure~\ref{intro1}. For $x,y\in\mathbb{N}$, denote the \emph{Hamming distance} $d_H$ of $x$ and $y$ by $d_H(x,y)=1$ if $x\neq y$ and $d_H(x,y)=0$ otherwise. Given a phylogenetic network $N$ of $\Gamma$ and, for a positive integer $p$, a function $C:\Gamma\to\{0,1,\dots,p\}$, called a \emph{character} of $\Gamma$ (over the set of \emph{states} $\{0,1,\dots,p\}$), the \emph{Softwired Parsimony Score Problem} (SPS) consists of finding a phylogenetic tree $T=(V,E)$ of $\Gamma$ isomorphic up to edge subdivisons to a subtree of $N$ \CorrB{such that the isomorphism is the identity map on $\Gamma$} and \CorrB{a map $C':V\to\{0,1,\dots,p\}$ with $C'(x)=C(x)$ for all $x\in\Gamma$, called an \emph{extension}} of $C$ to $V$, that minimizes
\begin{align*}
\CorrB{\text{score}(T,C')=\sum_{\{u,v\}\in E}d_H(C'(u),C'(v)).}
\end{align*}
It can be shown that the phylogenetic tree shown \CorrB{on} the right in Figure~\ref{intro1}, and the indicated extension, is an optimal solution to the SPS instance formed by the taxa $\Gamma$, phylogenetic network $N$ and character $C$ in the same figure.
 
\begin{figure}[pos=!t,align=\centering]
\centering
\includegraphics[scale=0.35]{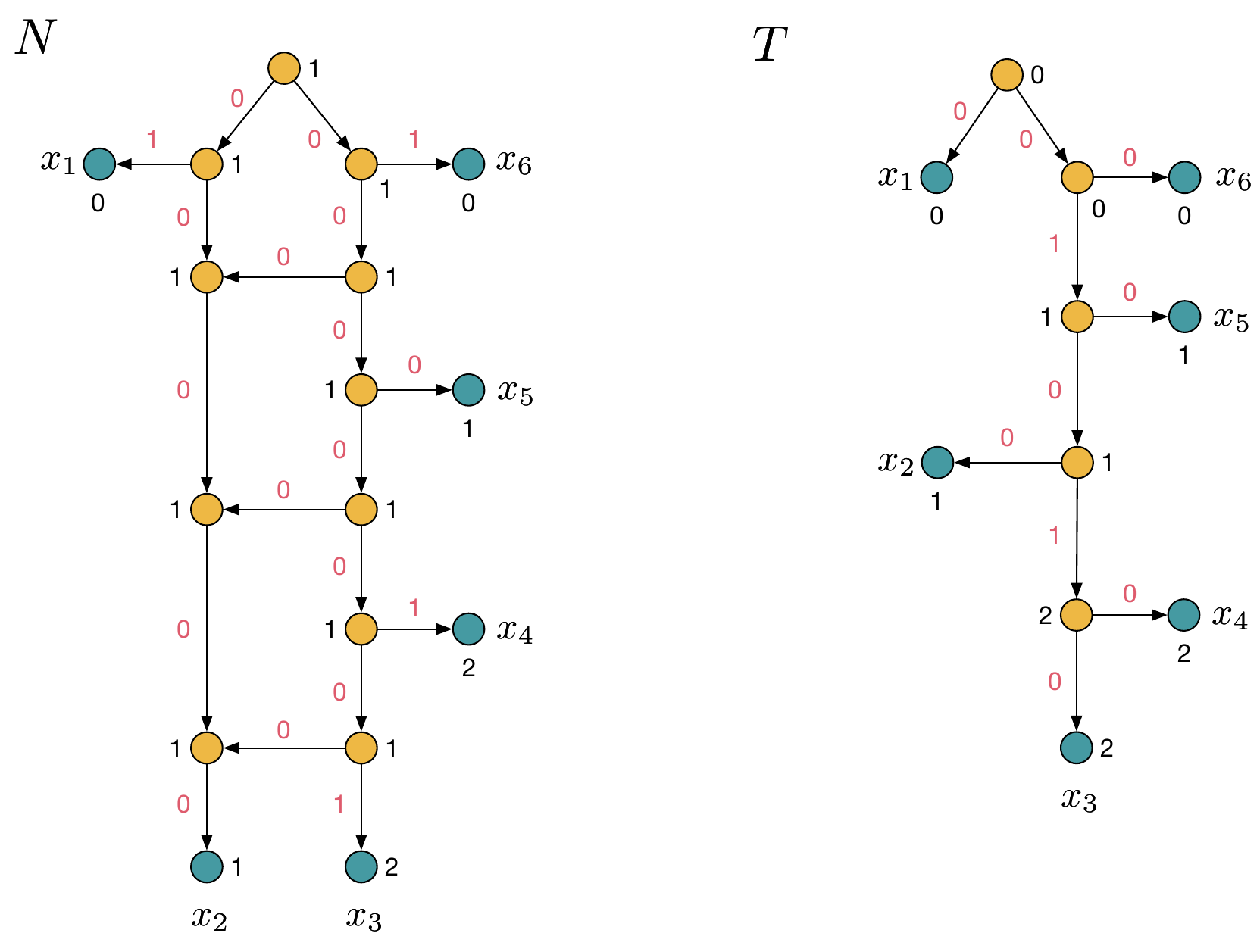}
\caption{Consider the set of taxa \CorrB{$\Gamma =\{x_1,x_2,\dots,x_6\}$} and the character $C:\Gamma\to\{0,1,2\}$ where $C(x_1)=C(x_6)=0$, $C(x_2)=C(x_5)=1$ and $C(x_3)=C(x_4)=2$. Then, the directed graph on the left (right) shows a phylogenetic network $N=(V,E)$ (tree \CorrB{$T=(W,F)$}) of $\Gamma$ for an extension $C'$ of $C$ to $V$ \CorrB{($W$) in black and distances $d_H(C'(u),C'(v))$ for all $(u,v)\in E$ ($(u,v)\in F$) in red. Observe that score$(T,C')=2$}. Hence, the SPS for $N$ is at at most 2 - in fact, \CorrB{we need at least two state changes because we have three different states}, so the SPS of $N$ is exactly 2. In contrast, the HPS for $N$ is 4, and this can be achieved by the extension $C'$ shown on $N$.}\label{intro1}
\end{figure}

The SPS has its origins in the estimation of phylogenetic trees under the maximum parsimony criterion~\citep{farris70,fitch71}. There the input is a sequence of DNA symbols for each taxon in $\Gamma$, which can be thought of as \CorrB{an (ordered) set} of characters, and the goal is to find the tree $T$ that minimizes the sum of scores, as defined above, \CorrB{with respect to all possible extensions and} ranging over all the characters in the set. An optimal phylogenetic tree estimated under this criterion might prove useful to explain hierarchical evolutionary relationships among taxa~\citep{steel03}. In practice this NP-hard optimization problem is tackled by heuristically searching through the space of trees, and for each such tree calculating its score. The calculation of the score of a given tree is often called the ``small parsimony'' problem. The small parsimony problem can easily be solved in polynomial time using Fitch's algorithm (i.e., Fitch's algorithm solves SPS on trees). However, phylogenetic trees cannot take into account reticulation events such as hybridizations or horizontal gene transfers, which play a significant role in the evolution of certain taxa~\citep{koonin01,bogart03,arnold97}. Therefore, one can extend the parsimony paradigm to networks displaying reticulation events, i.e., graphs which model the evolutionary process and allow for the existence of cycles. As for the inference of trees, the parsimony problem is then tackled in practice by searching through the space of networks, scoring each candidate network as it is found i.e. solving the small parsimony problem on the network.
 
 The objective function we consider in this article, SPS, is one of several different ways of scoring a phylogenetic network. Essentially, the SPS arises when we restrict our search space to one fixed phylogenetic network $N$ and we seek out the best phylogenetic tree $T$ `inside'~$N$. A class of phylogenetic networks which we will focus on in this article, and which have recently been intensively studied in the phylogenetic networks literature, are \emph{tree-child} networks~\citep{kong1}. This subclass is topologically restricted in a specific way which often makes NP-hard optimization problems on phylogenetic networks (comparatively) easier to solve, see e.g \cite{van2010locating,van2022practical}.
Unfortunately,  SPS remains challenging even when we restrict our attention to rooted, tree-child networks with an additional algorithmically advantageous topological restriction known as \emph{time consistency}: for any $\epsilon>0$ the \CorrB{best known} inapproximability factor is $|\Gamma|^{1-\epsilon}$. If we focus only on rooted, \emph{binary} phylogenetic networks that are not necessarily tree-child the \CorrB{best known} inapproximability
factor is $|\Gamma|^{\frac{1}{3}-\epsilon}$~\citep{fischer15}. Between these two variations is the situation when the class of networks we consider are rooted, binary and tree-child. 
\CorrB{\begin{defn}
$\mathcal{N}$ denotes the class of rooted, binary, tree-child phylogenetic networks of $\Gamma$.
\end{defn}}
The problem remains NP-hard \CorrB{on $\mathcal{N}$}, but how approximable is it? In this article we develop a polynomial-time 2-approximation for the SPS problem on this class; prior to this result no non-trivial approximation algorithms were known. In Section~\ref{sec2} we give some background on the fundamental properties of $\mathcal{N}$. In Section~\ref{sec3} we emulate Fitch's algorithm by a primal-dual algorithm. In Section~\ref{sec4} we extend the results from Section~\ref{sec3} to obtain a 2-approximation algorithm for the SPS restricted to networks from class $\mathcal{N}$ and any character of $\Gamma$. Furthermore, we show that the approximation guarantee of our algorithm is tight. In Section~\ref{sec:5} we reflect on the broader significance of our results.

Finally, we note for the sake of clarity that there is also a model in the literature where the score of a network on a given character  is defined by aggregating Hamming distances over \emph{all} the edges in the network, rather than just those belonging to a certain single tree inside the network. This is called the \emph{Hardwired Parsimony Score} (HPS) problem and it behaves rather differently to the SPS model that we consider. The HPS can be solved in polynomial time for characters $C:\Gamma\to\{0,1\}$, called \emph{binary} characters, and it is essentially a multiterminal cut problem for non-binary characters, yielding APX-hardness and constant-factor approximations \cite{fischer15}. HPS is also a very close relative of the ``happy vertices/edges'' problem that has recently received much attention in the algorithms literature \cite{zhang2018improved}.

\section{Notation and background}\label{sec2}
For a graph $G$ we denote $V(G)$ and $E(G)$ as the \CorrB{vertex set and edge set} of $G$, respectively. We call a connected graph $G$ \emph{rooted} if $G$ is directed and has a unique vertex $\rho\in V(G)$ having in-degree zero and there exists a directed path from $\rho$ to any vertex of in-degree 1 and out-degree 0 of $G$. We call $\rho$ the \emph{root} of $G$ and denote $d_{\text{in}}(v)$ and $d_{\text{out}}(v)$ as the in-degree and out-degree of the vertex $v\in V(G)$, respectively. In this article we only consider rooted directed acyclic graphs. We call a directed $G$ \emph{binary} if $d_{\text{in}}(v)+d_{\text{out}}(v)\leq 3$ for all $v\in V(G)$. For a rooted, binary graph $G$, we define
\begin{align*}
V_{\text{int}}^r&=\left\{v\in V(G)\,:\,d_{\text{in}}(v)=2, d_{\text{out}}(v)=1\right\},\\
V_{\text{int}}^t&=\left\{v\in V(G)\,:\,d_{\text{in}}(v)=1, d_{\text{out}}(v)=2\right\},\\
V_{\text{ext}}&=\left\{v\in V(G)\,:\,d_{\text{in}}(v)=1, d_{\text{out}}(v)=0\right\}.
\end{align*}
We call $V_{\text{int}}^r$, $V_{\text{int}}^t$ and $V_{\text{ext}}$ the \emph{reticulation vertices}, \emph{internal tree vertices} and \emph{leaves} of~$G$, respectively, and omit the mention of $G$ from the definition of these symbols to simplify our notation. When $G$ is a phylogenetic network, as introduced in the last section, then we assume there exist no vertices $v\in V(G)$ with $d_{\text{in}}(v)=d_{\text{out}}=1$. Hence, in this case we have $V(G)=V_{\text{int}}^r\cup V_{\text{int}}^t\cup V_{\text{ext}}$. Analogously, define
\begin{align*}
E_{\text{int}}^r=\left\{(u,v)\in E(G)\,:\,v\in V_{\text{int}}^r\right\},~E_{\text{int}}^t=\left\{(u,v)\in E(G)\,:\,v\in V_{\text{int}}^t\right\},~E_{\text{ext}}=\left\{(u,v)\in E(G)\,:\,v\in \CorrB{V_{\text{ext}}}\right\},
\end{align*}
i.e., $E(G)=E_{\text{int}}^r\cup E_{\text{int}}^t\cup E_{\text{ext}}$. We call $E_{\text{int}}^r$, $E_{\text{int}}^t$ and $E_{\text{ext}}$ the \emph{reticulation edges}, \emph{internal tree edges} and \emph{external tree edges} of~$G$, respectively. The graph associated with sets $V_{\text{int}}^r$, $V_{\text{int}}^t$, $V_{\text{ext}}$, $E_{\text{int}}^r$, $E_{\text{int}}^t$, $E_{\text{ext}}$ will be clear from the context.

We call a rooted phylogenetic network \emph{tree-child} if
\begin{align*}
\forall\,u\in V_{\text{int}}^r\cup V_{\text{int}}^t~\exists\,(u,v)\in E_{\text{int}}^t\cup E_{\text{ext}}.
\end{align*}
The network $N$ in Figure~\ref{intro1} is not tree-child. This can be certified by observing that there is a non-leaf vertex whose only child is a reticulation vertex. The network in Figure~\ref{intro2} is, however, tree-child. We say a phylogenetic tree $T$ is \emph{displayed} by a phylogenetic network $N$ if $T$ is a subtree of $N$ up to edge subdivisions. \CorrB{For example, the tree $T$ in Figure~\ref{intro1} is displayed by network $N$ in the same figure.} Hence, the SPS seeks to find a phylogenetic tree $T$ displayed by the given phylogenetic network $N$,  and an extension $C'$, such that score$(T,C')$ is minimum. Making a phylogenetic network $N$ acyclic by deleting exactly one reticulation edge $(u,v)$ for all reticulation vertices $v$ is called \emph{switching}. Let $\mathcal{T}(N)$ denote the set of all phylogenetic trees displayed by the phylogenetic network $N$ and $\mathcal{S}(N)$ the set of switchings of the network. An attractive property of tree-child networks is that a phylogenetic tree $T$ is in $\mathcal{T}(N)$ if and only if
there is a switching $S\in\mathcal{S}(N)$  that is isomorphic to $T$ up to edge subdivision. This property holds for tree-child networks $N$ because in every $S\in\mathcal{S}(N)$ and for all $v\in V(S)$ there exists a path in $S$ from $v$ to a leaf of $N$. Hence, in this case it is sufficient to determine a switching of $N$ to solve the SPS. In more general network classes switchings can also be used to characterize displayed trees, but there isomorphism up to edge subdivision does not necessarily hold: the switching might contain leaf vertices unlabelled by taxa, for example.

\begin{figure*}[pos=!t,align=\centering]
\centering
\includegraphics[scale=0.5]{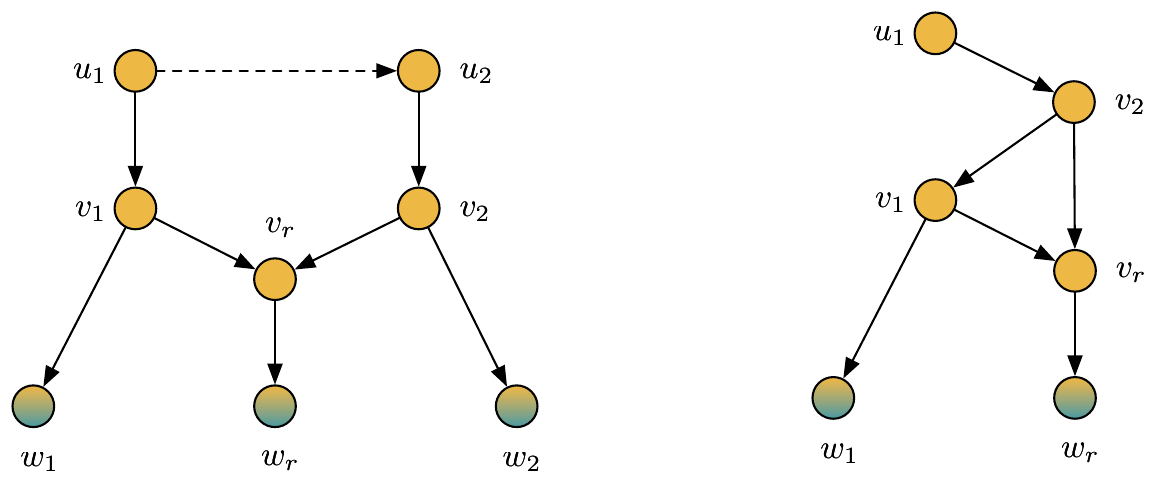}
\caption{Two possible subgraphs \CorrB{depicting parents and children of a reticulation vertex $v_r$ and their incidence relations} in a rooted, binary, tree-child network $N$. Vertices $w_1,w_2$ and $w_r$ can be internal vertices or leaves. Note that in the left subgraph $u_1 = u_2$ is possible.}\label{ret1}
\end{figure*}

\begin{figure*}[pos=!b,align=\centering]
\centering
\includegraphics[scale=0.5]{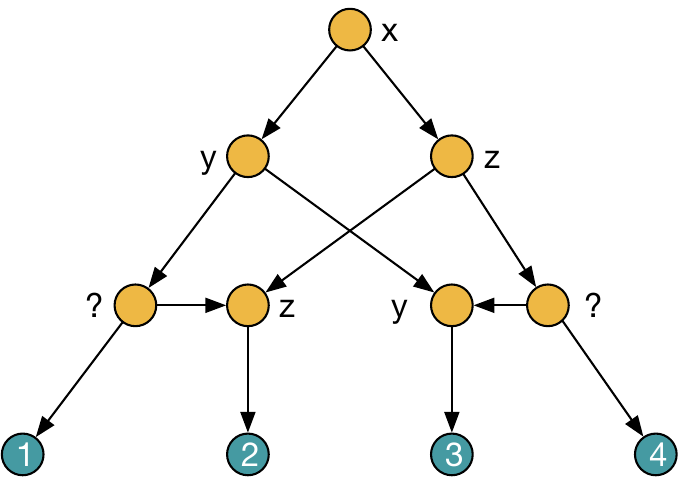}
\caption{Consider the set of taxa $\Gamma =\{1,2,3,4\}$. The graph shows a rooted, binary, phylogenetic network $N$ of $\Gamma$ in which the internal vertices are labeled by a function $t:V_{\text{int}}^t\cup V_{\text{int}}^r\to\{x,y,z,?\}$~\citep{baroni06}. The network is not time-consistent, but it is tree-child.} \label{intro2}
\end{figure*}

Moreover, for networks from the set $\mathcal{N}$ of rooted, binary, tree-child phylogenetic networks, the \CorrB{local structure of reticulation vertices} is shown in Figure~\ref{ret1} (when allowing $u_1=u_2$). We call a rooted phylogenetic network $N$ \emph{triangle-free} if every cycle in $N$ has length at least four. 
\CorrB{\begin{obs}\label{obs::triangle}
Let $N\in\mathcal{N}$ contain at least one triangle and let $N'$ be the phylogenetic network of $\Gamma$ obtained from $N$ by removing edges in triangles of $N$ and deleting resulting edge subdivisons. Then, for a character $C$ of $\Gamma$, the optimal objective function values of the SPS for $N$ and $C$, and the SPS for $N'$ and $C$ are equal.
\end{obs}}
To see Observation~\ref{obs::triangle}, consider the graphs in Figure~\ref{ret1}. We argue that the triangle on the right does not yield any benefit to the softwired parsimony score. On the one hand, we can remove edge $(v_2,v_r)\in E_{\text{int}}^r$. Then the resulting network is isomorphic up to edge subdivisions to a network $N_1$ in which $(u_1,v_1),(v_1,w_1)$ and $(v_1,w_r)$ form edges. On the other hand, we can remove reticulation edge $(v_1,r)$ in $N$. Then the resulting network is isomorphic up to edge subdivisions to a network $N_2$ in which $(u_1,v_2),(v_2,w_1)$ and $(v_2,w_r)$ form edges. By construction networks $N_1$ and $N_2$ have the same softwired parsimony score. Hence, the decision on whether to use edge $(v_1,v_r)$ or $(v_2,v_r)$ to construct a phylogenetic tree displayed by $N$ has no effect on the optimal solution to the SPS. We can apply this procedure to remove all triangles and resulting edge subdivisions. Thus, throughout this article we assume that networks $N\in\mathcal{N}$ are triangle-free.

We call a rooted phylogenetic network \emph{time-consistent} if there exists a function $t:V(N)\to \mathbb{N}$, called a \emph{time-stamp function}, such that
\begin{align*}
t(v)-t(u)=\begin{cases}0&\text{if $(u,v)\in E_{\text{int}}^r$},\\
\geq 1 &\text{if $(u,v)\in E_{\text{int}}^t\cup E_{\text{ext}}$}.
\end{cases}
\end{align*}
An example of a network $N$ which is not time-consistent is shown in Figure~\ref{intro2}. To see this, first observe that labeling the root by any natural number $x\in\mathbb{N}$ forces its children to be labeled by $y>x$ and $z>x$, respectively. This means, the fathers of taxa $2$ and $3$ have to be labeled as $z$ and $y$, respectively, because they are reticulation vertices. To label the remaining internal vertices we require $y<z$ and $z<y$, which is impossible. Thus, $N$ is not time-consistent. For a positive integer $p$, let $S_p=\{0,1,\dots,p\}$ be the set of \emph{character states} and let the power set $2^{S_p}$ of $S_p$ denote the set of \emph{character state sets}. With a little abuse of notation we write singleton character state sets $\{q\}\in 2^{S_p}$ interchangeably as the character state $q$. This allows us to define characters of $\Gamma$ by taking unions and intersections of character state sets. \CorrB{To this end, we call vertices with shared parent vertices \emph{siblings}.} Let $A$ be an algorithm for a phylogenetic network which iteratively propagates character states from the leaves to all reticulation vertices $v_r$ only if \CorrB{the siblings of $v_r$} have been assigned a character state $S\subseteq S_p$ by~$A$. Note that $A$ does not necessarily exist. If it does, we say that $A$ \CorrB{respects a \emph{reticulation consistent ordering}}.

\begin{algorithm}[!b]
 \KwIn{A set of taxa $\Gamma=\{1,2,\dots,n\}$, a rooted, binary phylogenetic tree $T$ of $\Gamma$; a character~$C$ of $\Gamma$; the root $\rho\in V(T)$}
 \KwOut{A character $C'$ of $V(T)$; minimum score$(T,C')$}
 \eIf{$\rho$ is a leaf}{
 	\Return{$C,0$}\;
 }{
 	$v_1,v_2\leftarrow$ children of $\rho$\;
 	\For{$i=1,2$}{
		$T[v_i]\leftarrow$ subtree of $T$ rooted in $v_i$\;
		$\Gamma[v_i]\leftarrow$ all taxa present in $T[v_i]$\;
		$D\leftarrow$ character $C$ restricted to $\Gamma[v_i]$\;
		$(D',s_i)\leftarrow$ Fitch($\Gamma[v_i]$,$T[v_i]$,$D$,$v_i$)\;
		\For{$w\in V(T[v_i])$}{
			$C'(w)\leftarrow D'(w)$\;
		}
	}
	\eIf{$C'(v_1)\cap C'(v_2)\neq\emptyset$}{
		$C'(\rho)\leftarrow C'(v_1)\cap C'(v_2)$\;
		\If{$\rho$ is the root with which the algorithm started}{
			Remove character states from $C'(\rho)$ at random until $C'(\rho)$ is a singleton\;
		}
		\If{$C'(\rho)$ is a singleton}{
			Recursively set $C'(v)=C'(\rho)$ or $C'(v)=s$ for the children $v$ of $\rho$ when $C'(\rho)\in C'(v)$ or $C'(\rho)\notin C'(v)$, $s\in C'(v)$, respectively\;
		}
		\Return{$C',s_1+s_2$}\;
	}{
		$C'(\rho)\leftarrow C'(v_1)\cup C'(v_2)$\;
		\Return{$C',s_1+s_2+1$}\;
	}
 }

 \caption{Fitch's algorithm}
 \label{algo::Fitch}
\end{algorithm}

\begin{obs}\label{obs::time}
If $N\in\mathcal{N}$ is time-consistent, \CorrB{then there exists an algorithm $A$ which respects a reticulation consistent ordering.}
\end{obs}

To see Observation~\ref{obs::time}, design an algorithm which iteratively propagates character states to vertices non-increasing in the evaluations of a time-stamp function $t$. Clearly, all children of a vertex $v$ which are internal tree vertices appear before $v$ in such a propagation sequence.

In the next section we consider the SPS where $N$ is a rooted, binary phylogenetic tree. In this case, the SPS can be solved in polynomial time using Fitch's algorithm~\citep{fitch71}. We construct a primal-dual algorithm $A$ for the SPS which emulates Fitch's algorithm. Subsequently, we show how to extend $A$ to phylogenetic networks $N\in\mathcal{N}$.

\section{An alternative proof of correctness for Fitch's algorithm}\label{sec3}
In this section, we consider the following special case of the SPS: given a set of taxa $\Gamma =\{1,2,\dots,n\}$, a rooted, binary phylogenetic tree $T$ of $\Gamma$ and a character $C$ of $\Gamma$, find an extension $C'$ of $C$ to $V(T)$ that minimizes score$(T,C')$. We call this problem the \emph{Binary Tree Parsimony Score Problem} (BTPS). It is well known that the BTPS can be solved in polynomial-time:

\begin{prop}\label{prop::fitch2}
The BTPS can be solved in polynomial time by Fitch's algorithm (see algorithm~\ref{algo::Fitch})~\citep{fitch71}.
\end{prop}

\CorrB{Notice that Fitch's algorithm solves the BTPS even if we remove the restriction to binary trees. Here, for the purpose of our analysis of the SPS in the next section, we do not discuss this aspect of the algorithm any further.} We construct an algorithm that operates like Algorithm~\ref{algo::Fitch} and prove that it solves the BTPS using linear programming duality arguments. This analysis will aid us in the next section in the development of an approximation algorithm for the SPS. Let $T$ be a rooted, binary phylogenetic tree of $\Gamma$, let $C:\Gamma\to S_p$ be a character of $\Gamma$ and let $C'$ be an extension of $C$ from $\Gamma$ to $V(T)$. \citet{fischer15} formulated an integer program to model the BTPS by introducing binary decision variables $x_v^s$ for all $v\in V_{\text{int}}^t$, $s\in S_p$, such that
\begin{align*}
x_v^s&=\begin{cases}1&\text{if $C'(v)=s$},\\
0&\text{otherwise},
\end{cases}
\end{align*}
and edge variables $c_e$ for all $e\in E$, such that
\begin{align*}
c_{(u,v)}&=d_H(C'(u),C'(v))=\begin{cases}1&\text{if $C'(u)\neq C'(v)$},\\
0&\text{otherwise}.
\end{cases}
\end{align*}
Then, given character states $C(v)$, $v\in V_{\text{ext}}$, the following program is an IP formulation of the BTPS:
\begin{form}\label{form::btps}
\begin{align}
\min~~\sum_{e\in E(T)}c_e~~~~~~~~~&\\
\text{s.t.}\,~~~~~~c_e-x_u^s+x_v^s&\geq 0~&~&\forall\,e=(u,v)\in E(T),s\in S_p\label{1::con1}\\
c_e+x_u^s-x_v^s&\geq 0~&~&\forall\,e=(u,v)\in E(T),s\in S_p\label{1::con2}\\
\sum_{s=0}^px_v^s&= 1~&~&\forall\,v\in V(T)\label{1::con3}\\
x_v^{C(v)}&=1~&~&\forall\,v\in V_{\text{ext}}\label{1::con4}\\
c_e&\in\{0,1\}~&~&\forall\,e\in E(T)\\
x_v^s&\in\{0,1\}~&~&\forall\,v\in V(T),s\in S_p
\end{align}
\end{form}
Next, consider the dual of the LP relaxation of Formulation~\ref{form::btps}:
\begin{form}\label{form::btpsDual}
\begin{align}
\max~~\sum_{v\in V(T)}\mu_v+\sum_{v\in V_{\text{ext}}}\eta_v\,~~~~~~~~~~~~~~~~~~~~~~~~~~~~~~~~~~~~~~~~~~~~~~~~~~~~~~~~~~~~~&\label{2::obj}\\
\text{s.t.}~~~~~~~~~~~~~~~~~~~~~~~~~~~~~~~~~~~~~~~~~~~~~~~~~~~~~~~~~~~~~~~~~~~~~~\sum_{s=0}^p\left(\beta_{u,v}^s+\beta_{v,u}^s\right)&\leq 1~&~&\forall\,(u,v)\in E(T)\label{2::con1}\\
\mu_v+\sum_{(v,w)\in E(T)}\left(\beta_{w,v}^s-\beta_{v,w}^s\right)+\sum_{(u,v)\in E(T)}\left(\beta_{u,v}^s-\beta_{v,u}^s\right)&\leq 0~&~&\forall\,v\in V_{\text{int}}^t,s\in S_p\label{con1::mu1}\\
\mu_v+\sum_{(u,v)\in E(T)}\left(\beta_{u,v}^s-\beta_{v,u}^s\right)&\leq 0~&~&\forall\,v\in V_{\text{ext}},s\in S_p\setminus\{C(v)\}\label{con1::mu2}\\
\mu_v+\eta_v+\sum_{(u,v)\in E(T)}\left(\beta_{u,v}^s-\beta_{v,u}^s\right)&\leq 0~&~&\forall\,v\in V_{\text{ext}},s=C(v)\label{con1::mu+eta}\\
\beta_{u,v}^s,\beta_{v,u}^s&\geq 0~&~&\forall\,(u,v)\in E(T),s\in S_p\\
\mu_v&\in\mathbb{R}~&~&\forall\,v\in V(T)\\
\eta_v&\in\mathbb{R}~&~&\forall\,v\in V_{\text{ext}}
\end{align}
\end{form}
\CorrB{Notice that dual variables $\beta_{u,v}^s$, $\beta_{v,u}^s$, $\mu_v$ and $\eta_v$ correspond to the primal constraints of form~\eqref{1::con1},~\eqref{1::con2},~\eqref{1::con3} and~\eqref{1::con4}, respectively. Conversely, the primal variables $c_e$ and $x_v^s$ correspond to the dual constraints of form~\eqref{2::con1} and~\eqref{con1::mu1},~\eqref{con1::mu2},~\eqref{con1::mu+eta}, respectively. Since constraints~\eqref{1::con4} depend only on both external vertices $V_{\text{ext}}$ and the character $C$, dual constraints~\eqref{con1::mu2} and~\eqref{con1::mu+eta} differentiate this structure imposed on primal variables $x_v^s$ from constraints~\eqref{con1::mu1}. Furthermore, observe that we do not introduce bounded dual variables for the primal LP relaxation constraints $c_e\leq 1$ and $x_v^s\leq 1$ because they produce an LP dual equivalent to Formulation~\ref{form::btpsDual} and introduce unnecessary degrees of freedom in the construction of optimal solutions to the BTPS. Indeed, imposing the primal LP relaxation constraints adds pairwise independent non-positive numbers to the lefthand side of constraints~\eqref{2::con1} to~\eqref{con1::mu+eta} and the objective function~\eqref{2::obj}.}

Formulation~\ref{form::btpsDual} can be written in a simpler form when we take the particular structure of the rooted, binary phylogenetic tree $T$ into account. Observe that every vertex in $V_{\text{int}}^t$ appears twice as $u$ for $(u,v)\in E(T)$ and once as $v$ for $(u,v)\in E(T)$. Hence, for $v\in V_{\text{int}}^t$, $s\in S_p$ with $(u,v),(v,w_1),(v,w_2)\in E(T)$, constraint~\eqref{con1::mu1} is of the form
\begin{align}
\mu_v+\beta_{w_1,v}^s-\beta_{v,w_1}^s+\beta_{w_2,v}^s-\beta_{vw_2}^s+\beta_{u,v}^s-\beta_{v,u}^s\leq 0\label{con::bbb}
\end{align}
and, for $v\in V_{\text{ext}}$, $s\in S_p$ with $(u,v)\in E(T)$, constraints~\eqref{con1::mu2} and~\eqref{con1::mu+eta} are of the form
\begin{align}
\mu_v+\beta_{u,v}^s-\beta_{v,u}^s&\leq 0\label{con::00b}\\
\text{and}~~\mu_v+\eta_v+\beta_{u,v}^s-\beta_{v,u}^s&\leq 0,\label{con::e00b}
\end{align}
respectively. In the following paragraphs we introduce all of the components that we will need to construct a primal-dual algorithm emulating Algorithm~\ref{algo::Fitch}. To this end, we will make frequent use of the complementary slackness conditions associated with the LP relaxation of Formulation~\ref{form::btps} and Formulation~\ref{form::btpsDual}:
\begin{align}
\left(c_e-x_u^s+x_v^s\right)\beta_{u,v}^s&=0~&~&\forall\,e=(u,v)\in E(T),s\in S_p\label{cs::dv1}\\
\left(c_e+x_u^s-x_v^s\right)\beta_{v,u}^s&=0~&~&\forall\,e=(u,v)\in E(T),s\in S_p\label{cs::dv2}\\
\left(\sum_{s=0}^p\left(\beta_{u,v}^s+\beta_{v,u}^s\right)-1\right)c_e&=0~&~&\forall\,e=(u,v)\in E(T)\label{cs::pv1}\\
\left(\mu_v+\beta_{w_1,v}^s-\beta_{v,w_1}^s+\beta_{w_2,v}^s-\beta_{vw_2}^s+\beta_{u,v}^s-\beta_{v,u}^s\right)x_v^s&=0~&~&\forall\,v\in V_{\text{int}}^t,s\in S_p\label{cs::pv2}\\
\left(\sum_{s=0}^px_v^s-1\right)\mu_v&=0~&~&\forall\,v\in V(T)\label{cs::dv3}\\
\left(x_v^{C(v)}-1\right)\eta_v&=0~&~&\forall\,v\in V_{\text{ext}}\label{cs::dv4}\\
\left(\mu_v+\beta_{u,v}^s-\beta_{v,u}^s\right)x_v^s&=0~&~&\forall\,v\in V_{\text{ext}},s\in S_p\setminus\{C(v)\}\label{cs::pv3}\\
\left(\mu_v+\eta_v+\beta_{u,v}^s-\beta_{v,u}^s\right)x_v^s&=0~&~&\forall\,v\in V_{\text{ext}},s=C(v)\label{cs::pv4}
\end{align}
For specific assumptions on the values of dual variables we provide some sufficient conditions for the local optimality of our primal and dual variables derived from complementary slackness that we will use throughout this article to prove the correctness of our algorithms. For a phylogenetic tree $T$ of $\Gamma$ and a character $C:\Gamma\to S_p$, with a little abuse of notation we say that a function $C':V(T)\to 2^{S_p}$ is an extension of $C$ from $\Gamma$ to $V(T)$. Such an extension $C'$ coincides with the notion of an extension in the definition of the SPS only if $C'(v)$ is a singleton character state for all $v\in V(T)$. \CorrB{Otherwise, we call $C'$ a \emph{set-extension} of $C$ from $\Gamma$ to $V(T)$.} It will be useful throughout this article to think of extensions of characters $C$ of $\Gamma$ as constructed by taking unions and intersections of the singleton character states of $\Gamma$, which will become more apparent in the next section.

\begin{prop}\label{prop::dualstep}
Let $\Gamma$ be a set of taxa and let $T$ be a rooted, binary phylogenetic tree of $\Gamma$ with root $\rho\in V(T)$. Let $C$ be a character of $\Gamma$ and \CorrB{let $C'$ be set-extension of $C$ from $\Gamma$ to $V(T)$}. Let $v\in V(T)$, $v\neq\rho$, with $(u,v),(v,w_1),(v,w_2)\in E(T)$, $w_1\neq w_2$, for $v\in V_{\text{int}}^t$ and $(u,v)\in E(T)$ for $v\in V_{\text{ext}}$. Assume $C'(w_1)\neq\emptyset\neq C'(w_2)$ and, for $v\in V_{\text{int}}^t$, $i=1,2$, $s\in C'(w_i)$, we have $\beta_{w_i,v}^s=1/|C'(w_i)|$.
\begin{enumerate}
\item Let $v\in V_{\text{int}}^t$, $C'(v)=C'(w_1)\cap C'(w_2)$. 
Then, we satisfy complementary slackness conditions~\eqref{cs::dv1} to~\eqref{cs::dv3} for $e=(u,v)$ and $v$ if one of the following two scenarios holds:
\begin{enumerate}[(i)]
\item $C'(v)=C'(w_1)=C'(w_2)=s$, $C'(u)=s'\neq s$ and $\beta_{v,u}^s=1$, $\beta_{v,u}^t=0$, $t\neq s$, $\beta_{u,v}^t=0$, $t\in S_p$, $\mu_v=-1$. 
\item $C'(v)=C'(w_1)=C'(w_2)=C'(u)$ and $\beta_{v,u}^s=1/|C'(v)|$, $s\in C'(v)$, $\beta_{v,u}^t=0$, $t\in S_p\setminus C'(v)$, $\beta_{u,v}^t=0$, $t\in S_p$, $\mu_v=-1/|C'(v)|$. 
\end{enumerate}
\item Let $v\in V_{\text{ext}}$ and $C'(v)=C(v)$. Then, we satisfy complementary slackness conditions~\eqref{cs::dv1} to~\eqref{cs::pv1} and \eqref{cs::dv3} to~\eqref{cs::pv4} for $e=(u,v)$ and $v$ if $\beta_{v,u}^{C(v)}=1$, $\mu_v=0$ and $\eta_v=1$.
\item Let $v\in V_{\text{int}}^t$, $C'(v)\nsubseteq C'(w_1)\cup C'(w_2)$. Then, there exist no values for primal and dual variables satisfying complementary slackness conditions~\eqref{cs::dv1} to~\eqref{cs::dv3} for $e=(u,v)$ and $v$.
\item Let $v\in V_{\text{int}}^t$, $C'(v)\subseteq C'(w_1)\cup C'(w_2)$, $C'(v)\cap C'(w_1)\neq\emptyset$ and $C'(v)\cap C'(w_1)\cap C'(w_2)=\emptyset$. 
\begin{enumerate}[(i)]
\item $C'(v)=s$, $C'(u)=t\neq s$. Then, there exist no values for primal and dual variables satisfying complementary slackness conditions~\eqref{cs::dv1} to~\eqref{cs::dv3} for $e=(u,v)$ and $v$.
\item $|C'(w_1)|=|C'(w_2)|$ and $\beta_{v,u}^s=1/|C'(v)|$, $\mu_v=1/|C'(v)|-1/|C'(w_1)|$. Then, we satisfy complementary slackness conditions~\eqref{cs::dv1} to~\eqref{cs::dv3} for $e=(u,v)$ and $v$.
\end{enumerate}
\end{enumerate}
\end{prop}
\begin{proof} Throughout this proof we always choose $x_v^s=1/|C'(v)|$ for $v\in V(T)$, $s\in S_p$.\\~\\
\underline{1.:} For $s\in C'(v)$, let $\beta_{v,u}^s=1/|C'(v)|$, i.e., condition~\eqref{cs::pv1} is satisfied. Let $C'(u)\cap C'(v)=\emptyset$. Then, $c_{(u,v)}=1$. Hence, for $s\in C'(v)$, $c_{(u,v)}+x_u^s-x_v^s=0$ only if $x_v^s=1-x_u^s$. Moreover, for $t\in C'(u)$, $c_{(u,v)}-x_u^t+x_v^t=0$ only if $x_u^t=1-x_v^t$. Then, fixing $C'(v)=s$ and $C'(u)=s'\neq s$, conditions~\eqref{cs::dv1}, \eqref{cs::dv2} and~\eqref{cs::dv3} are satisfied. Then, conditions~\eqref{cs::pv2} are equivalent to
\begin{align*}
\mu_v+\frac{1}{|C'(w_1)|}-0+\frac{1}{|C'(w_2)|}-0+0-1 &\leq 0,\\
\mu_v+\frac{1}{|C'(w_1)|}-0+0-0+0 -0&\leq 0~&~&\forall\,t\in C'(w_1)\setminus C'(w_2),\\
\mu_v+0-0+\frac{1}{|C'(w_2)|}-0+0 -0&\leq 0~&~&\forall\,t\in C'(w_2)\setminus C'(w_1),\\
\mu_v+0-0+0-0+0 -0&\leq 0~&~&\forall\,t\in S_p\setminus\left(C'(w_2)\cup C'(w_1)\right).
\end{align*}
Since $x_v^s=1$, we require $\mu_v=1-1/|C'(w_1)|-1/|C'(w_2)|$. Hence, conditions~\eqref{cs::pv2} are satisfied only if $C'(w_1)=C'(w_2)$. Since $C'(v)=C'(w_1)\cap C'(w_2)$, we have $C'(v)=C'(w_1)=C'(w_2)$.

Next, let $C'(u)\cap C'(v)\neq\emptyset$. Then, $c_{(u,v)}=0$ and, for $s\in S_p$, $c_{(u,v)}+x_u^s-x_v^s=0$ only if $x_u^s=x_v^s$. Hence, fixing $C'(v)=C'(u)$, conditions~\eqref{cs::dv1}, \eqref{cs::dv2} and~\eqref{cs::dv3} are satisfied. Then, conditions~\eqref{cs::pv2} are equivalent to
\begin{align*}
\mu_v+\frac{1}{|C'(w_1)|}-0+\frac{1}{|C'(w_2)|}-0+0-\frac{1}{|C'(v)|} &\leq 0~&~&\forall\,s\in C'(v),\\
\mu_v+\frac{1}{|C'(w_1)|}-0+0-0+0 -0&\leq 0~&~&\forall\,t\in C'(w_1)\setminus C'(w_2),\\
\mu_v+0-0+\frac{1}{|C'(w_2)|}-0+0 -0&\leq 0~&~&\forall\,t\in C'(w_2)\setminus C'(w_1),\\
\mu_v+0-0+0-0+0 -0&\leq 0~&~&\forall\,t\in S_p\setminus\left(C'(w_1)\cup C'(w_2)\right).
\end{align*}
Since $x_v^s=1/|C'(v)|$, $s\in C'(v)$, we require $\mu_v=1/|C'(v)|-1/|C'(w_1)|-1/|C'(w_2)|$. Hence, conditions~\eqref{cs::pv2} are satisfied only if $1/|C'(v)|\leq\min\{1/|C'(w_1)|,1/|C'(w_2)|\}$. Equivalently, $|C'(v)|\geq\max\{|C'(w_1)|,|C'(w_2)|\}$. Since $C'(v)=C'(w_1)\cap C'(w_2)$, we arrive at $C'(v)=C'(w_1)=C'(w_2)$.\\~\\
\underline{2.:} Let $C'(u)\cap C'(v)=\emptyset$, i.e., $C(v)\notin C'(u)$. Then, analogous to the construction in the proof of Proposition~\ref{prop::dualstep}.1, fixing $C'(v)=C(v)$ and $C'(u)=t\neq C(v)$ gives us a construction of dual variables such that conditions~\eqref{cs::dv1} to~\eqref{cs::pv1} and~\eqref{cs::dv4} are satisfied, and conditions~\eqref{cs::pv3} and~\eqref{cs::pv4} can be equivalently stated as
\begin{align*}
\mu_v+\eta_v +0-1 &=0,~~&~\mu_v +0-0&\leq 0~~~~\forall\,t\in S_p\setminus C'(v).
\end{align*}
These conditions are satisfied for $\eta =1$ and $\mu=0$.

Next, let $C'(u)\cap C'(v)\neq\emptyset$, i.e., $C(v)\in C'(u)$. Then, we can draw analogous conclusions because $|C'(v)|=1$.\\~\\
\underline{3.:} Without loss of generality $C'(u)\cap C'(v)=\emptyset$. Then, analogous to the construction in the proof of Proposition~\ref{prop::dualstep}.1, fixing $C'(v)=s$, $C'(u)=s'\neq s$, gives us a construction of dual variables such that conditions~\eqref{cs::dv1} to~\eqref{cs::pv1} and~\eqref{cs::dv3} are satisfied, and  conditions~\eqref{cs::pv2} can be equivalently stated as
\begin{align*}
\mu_v+0-0+0-0+0-1 &=0\\
\mu_v+\frac{1}{|C'(w_1)|}-0+\frac{1}{|C'(w_2)|}-0+0-0&\leq 0~&~&\forall\,t\in C'(w_1)\cap C'(w_2)\\
\mu_v+\frac{1}{|C'(w_1)|}-0+0-0+0-0&\leq 0~&~&\forall\,t\in C'(w_1)\setminus C'(w_2)\\
\mu_v+0-0+\frac{1}{|C'(w_2)|}-0+0-0&\leq 0~&~&\forall\,t\in C'(w_2)\setminus C'(w_1)\\
\mu_v+0-0+0-0+0-0&\leq 0~&~&\forall\,t\in S_p\setminus\left(C'(v)\cup C'(w_1)\cup C'(w_2)\right)
\end{align*}
These conditions are satisfied only if $C'(w_1)=C'(w_2)=\emptyset$ which is impossible by assumption. Recall from the proof of Proposition~\ref{prop::dualstep}.1 that our construction of dual variables is unique because any other choice for character state sets $C'(v)$ and $C'(u)$ does not satisfy conditions~\eqref{cs::dv1} to~\eqref{cs::pv1} and~\eqref{cs::dv3}.\\~\\
\underline{4.:} Let $C'(u)\cap C'(v)=\emptyset$. Without loss of generality $C'(v)\cap C'(w_2)=\emptyset$. Then, we can draw analogous conclusions to the proof of Proposition~\ref{prop::dualstep}.3 by fixing $C'(v)=s$, $C'(u)=s'\neq s$, i.e., $s\in C'(w_1)$, and for conditions~\eqref{cs::pv2} defined by
\begin{align*}
\mu_v+\frac{1}{|C'(w_1)|}-0+0-0+0-1 &=0\\
\mu_v+\frac{1}{|C'(w_1)|}-0+\frac{1}{|C'(w_2)|}-0+0-0&\leq 0~&~&\forall\,t\in C'(w_1)\cap C'(w_2)\\
\mu_v+\frac{1}{|C'(w_1)|}-0+0-0+0-0&\leq 0~&~&\forall\,t\in C'(w_1)\setminus C'(w_2),t\neq s\\
\mu_v+0-0+\frac{1}{|C'(w_2)|}-0+0-0&\leq 0~&~&\forall\,t\in C'(w_2)\setminus C'(w_1)\\
\mu_v+0-0+0-0+0-0&\leq 0~&~&\forall\,t\in S_p\setminus\left(C'(w_1)\cup C'(w_2)\right).
\end{align*}

Next, let $C'(u)\cap C'(v)\neq\emptyset$. Then, analogous to the construction in the proof of Proposition~\ref{prop::dualstep}.1, fixing $C'(v)=C'(u)$, conditions~\eqref{cs::dv1} to~\eqref{cs::dv3} and~\eqref{cs::pv1} are satisfied, and, using our assumptions on sets $C'(w_1)$ and $C'(w_2)$, conditions~\eqref{cs::pv2} are equivalent to
\begin{align}
\mu_v+0-0+0-0+0-\frac{1}{|C'(v)|} &= 0~&~&\forall\,s\in C'(v)\setminus\left(C'(w_1)\cup C'(w_2)\right),\label{4eq1}\\
\mu_v+\frac{1}{|C'(w_1)|}-0+0-0+0-\frac{1}{|C'(v)|} &= 0~&~&\forall\,s\in C'(v)\cap C'(w_1),\label{4eq2}\\
\mu_v+0-0+\frac{1}{|C'(w_2)|}-0+0-\frac{1}{|C'(v)|} &= 0~&~&\forall\,s\in C'(v)\cap C'(w_2),\label{4eq3}\\
\mu_v+\frac{1}{|C'(w_1)|}-0+0-0+0 -0&\leq 0~&~&\forall\,t\in C'(w_1)\setminus\left(C'(v)\cup C'(w_2)\right),\label{4ineq1}\\
\mu_v+\frac{1}{|C'(w_1)|}-0+\frac{1}{|C'(w_2)|}-0+0 -0&\leq 0~&~&\forall\,t\in C'(w_1)\cap C'(w_2),\label{4ineq2}\\
\mu_v+0-0+\frac{1}{|C'(w_2)|}-0+0 -0&\leq 0~&~&\forall\,t\in C'(w_2)\setminus\left(C'(v)\cup C'(w_1)\right),\label{4ineq3}\\
\mu_v+0-0+0-0+0 -0&\leq 0~&~&\forall\,t\in S_p\setminus\left(C'(v)\cup C'(w_1)\cup C'(w_2)\right)\label{4ineq4}.
\end{align}
Since $C'(v)\cap C'(w_1)\neq\emptyset$, equations~\eqref{4eq1} and~\eqref{4eq3} require that $\mu_v=1/|C'(v)|-1/|C'(w_1)|$, $C'(v)\subseteq C'(w_1)\cup C'(w_2)$ and either $C'(v)\cap C'(w_2)=\emptyset$ or $|C'(w_1)|=|C'(w_2)|$. This means, inequalities~\eqref{4ineq1} and~\eqref{4ineq2} do not hold, i.e., we require $C'(w_1)\subseteq C'(v)\cup C'(w_2)$ and $C'(w_1)\cap C'(w_2)=\emptyset$. This means, either $C'(w_1)=C'(v)$ or $|C'(w_1)|=|C'(w_2)|$. Furthermore, from inequalities~\eqref{4ineq3} we infer that $C'(w_2)\subseteq C'(v)$. This is not possible when $C'(v)=C'(w_1)$. Thus, we require $|C'(w_1)|=|C'(w_2)|$.
\end{proof}

Conditions in Proposition~\ref{prop::dualstep} hold similarly for the root of a rooted, binary phylogenetic tree:

\begin{cor}\label{prop::dualstep::root}
Let $\Gamma$ be a set of taxa and let $T$ be a rooted, binary phylogenetic tree of $\Gamma$ with root $v\in V(T)$. Let $C$ be a character of $\Gamma$ and \CorrB{let $C'$ be an set-extension of $C$ from $\Gamma$ to $V(T)$.} Let $(v,w_1),(v,w_2)\in E(T)$, $w_1\neq w_2$. Assume $C'(w_1)\neq\emptyset\neq C'(w_2)$ and, for $i=1,2$, $s\in C'(w_i)$, we have $\beta_{w_i,v}^s=1/|C'(w_i)|$. Then, Propositions~\ref{prop::dualstep}.1 to~\ref{prop::dualstep}.4 hold for the same choices of variable values except for variable $\mu_v$ which needs to be substituted by $\mu_v-1/|C'(v)|$.
\end{cor}

\CorrB{Now, we will outline a primal-dual scheme for the BTPS (see algorithm~\ref{PD-HPS-trees}). We will follow the same rationale as algorithm~\ref{algo::Fitch}: recursively propagate the character states on leaves to the root by taking an intersection of character states of siblings $v_1,v_2$ if $C'(v_1)\cap C'(v_2)\neq\emptyset$ (Dual step (i)) and the union of character states otherwise (Dual step (ii)). Subsequently we translate these decisions for the dual to an (locally) optimal solution of the primal in the Primal step using Proposition~\ref{prop::dualstep} and Corollary~\ref{prop::dualstep::root} which includes the "clean-up phase" of Fitch's algorithm, i.e., the removal of redundant character states.} To this end, consider a character $C:\Gamma\to S_p$ and let $s\in S_p$ such that $\left|\{v\in\Gamma\,:\,C(v)=s\}\right|$ is maximum.
\begin{enumerate}
\item Choose an initial feasible primal solution: set $C'(v)=s$, $x_v^s=1$ for all $v\in V_{\text{int}}$ and $c_{(u,v)}=|x_u^s-x_v^s|$ for all $(u,v)\in E(T)$.
\item Choose an initial infeasible dual solution: set $\beta_{u,v}^t=\frac{1}{|S_p|}$, $\beta_{v,u}^t=0$, $(u,v)\in E(T)$, $t\in S_p$, $\mu_v=0$, $v\in V(T)$, and $\eta_v=1$, $v\in V_{\text{ext}}$.
\item Dual step: consider a violated dual constraint for $v\in V(T)$ of form
\CorrB{\begin{align*}
\mu_v+\eta_v+\sum_{(u,v)\in E(T)}\left(\beta_{u,v}^s-\beta_{v,u}^s\right)&\leq 0~&~&\forall\,v\in V_{\text{ext}},s=C(v)
\end{align*}
or
\begin{align*}
\mu_v+\sum_{(v,w)\in E(T)}\left(\beta_{w,v}^s-\beta_{v,w}^s\right)+\sum_{(u,v)\in E(T)}\left(\beta_{u,v}^s-\beta_{v,u}^s\right)&\leq 0~&~&\forall\,v\in V_{\text{int}}^t,s\in S_p.
\end{align*}}In the first case, for $(u,v)\in E(T)$, set $\beta_{v,u}^{C(v)}=1$, $\beta_{v,u}^t=0$ for all $t\in S_p\setminus\{C(v)\}$ and $\beta_{u,v}^t=0$ for all $t\in S_p$. In the latter case, for $(u,v),(v,w_1),(v,w_2)\in E(T)$, \CorrB{if}
\begin{enumerate}[(i)]
\item $C'(w_1)\cap C'(w_2)\neq\emptyset$: set $C'(v)=C'(w_1)\cap C'(w_2)$, $\beta_{v,u}^s=\frac{1}{|C'(v)|}$ for all $s\in C'(v)$, $\beta_{v,u}^s=0$ for all $s\in S_p\setminus C'(v)$, $\beta_{u,v}^t=0$ for all $t\in S_p$, and, if $v$ is not the root, $\mu_v=-\frac{1}{|C'(v)|}$. If $v$ is the root, then $\mu_v=-\frac{2}{|C'(v)|}$. Thereafter, set $C'(w_1)=C'(w_2)=C'(v)$.
\item $C'(w_1)\cap C'(w_2)=\emptyset$: set $C'(v)=C'(w_1)\cup C'(w_2)$, $\beta_{v,u}^s=\frac{1}{|C'(v)|}$, $s\in C'(v)$, $\beta_{v,u}^s=0$, $s\in S_p\setminus C'(v)$, $\beta_{u,v}^t=0$, $t\in S_p$, and, if $v$ is not the root, $\mu_v=\frac{1}{|C'(v)|}-\frac{1}{|C'(w_1)|}$. If $v$ is the root, then $\mu_v=-\frac{1}{|C'(w_1)|}$.
\end{enumerate}
\item Primal step: set $x_v^s=1$, $s\in C'(v)$, $c_{(u,v)}=|C'(u)\setminus C'(v)|$ and, for $i=1,2$, $c_{(v,w_i)}=|C'(v)\setminus C'(w_i)|$. For $i=1,2$, if $C'(w_i)$ was changed in the latest dual step, then recursively remove character states of the children of $C'(w_i)$ which are not in $C'(w_i)$ until no more character states can be removed or the recursion arrives at a singleton character state set.
\end{enumerate}

\begin{algorithm}[!t]
 \KwIn{A set of taxa $\Gamma=\{1,2,\dots,n\}$, a rooted, binary phylogenetic tree $T$ of $\Gamma$; a character $C:\Gamma\to S_p$; $s\in S_p$ such that $\left|\{v\in\Gamma\,:\,C(v)=s\}\right|$ is maximum}
 \KwOut{A character $C'$ of $V(T)$; minimum score$(T,C')$}
 Choose an initial feasible primal solution\;
 Choose an initial infeasible dual solution\;
 \While{There exists a violated dual constraint}{
 	Apply a dual step for $v\in V(T)$ such that there exist no violated dual constraints associated with the children of $v$. This results in a propagation of character states to $v$ from its children via a set intersection or union\;
	Apply a primal step to establish local optimality of dual and primal variables in vertex $v$\;
 }
 Fix $C'(\rho)$ to be a singleton\;
 Apply a primal step\;
 \Return{$(C',\sum_{e\in E}c_e)$}
 \caption{A Primal-Dual Algorithm for BTPS}
 \label{PD-HPS-trees}
\end{algorithm}

\begin{prop}\label{prop::treeOpt}
The BTPS can be solved in polynomial time by Algorithm~\ref{PD-HPS-trees}.
\end{prop}
\begin{proof}
First, observe that the dual step is applied to an internal vertex $v$ if and only if, the children $w_1$ and $w_2$ have been processed by a previous dual step. After one dual step for some $v\in V(T)$ no additional dual constraint is violated. Moreover, for vertex $v$, all constraints of the form~\eqref{con1::mu1} or all constraints of form~\eqref{con1::mu2} and~\eqref{con1::mu+eta} are no longer violated. 

Now, observe that after the application of dual step (i) we satisfy the complementary slackness conditions for vertex $v$ locally due to Proposition~\ref{prop::dualstep}.1. Whenever a change of sets $C'(w_1)$ and $C'(w_2)$ in dual step (i) occurs, we have to argue that the recursion in the subsequent primal step does not prohibit global optimality. Indeed, for $i=1,2$, if $w_i$ fits into Proposition~\ref{prop::dualstep}.1 as vertex $v$ by changing the character state sets of the children of $w_i$, then $w_i$ can be locally optimal for an appropriate change of variables as described in Proposition~\ref{prop::dualstep}.1. However, applying the same changes to $w_i$ fitting into Proposition~\ref{prop::dualstep}.4 might require a change of variables not captured by Proposition~\ref{prop::dualstep} and character state sets to maintain local optimality because of the condition $|C'(w_1)|=|C'(w_2)|$. Since this requirement becomes redundant for singleton character state sets, we can conclude again that after the termination of Algorithm~\ref{PD-HPS-trees} we can maintain local optimality for $w_i$ by an appropriate change of variables as described in Proposition~\ref{prop::dualstep}.4 (again $w_i$ is seen as vertex $v$ in the application of the proposition). With analogous arguments we can deduce that dual step (ii) does not prohibit global optimality either. Thus, in total we conclude that by setting $C'(\rho)$ as a singleton Algorithm~\ref{PD-HPS-trees} constructs a feasible primal and feasible dual solution which satisfy all complementary slackness conditions~\ref{cs::dv1} to~\ref{cs::pv4}.
\end{proof}

In the next section we expand our analysis of the SPS from rooted, binary phylogenetic trees to the class $\mathcal{N}$ of rooted, binary, tree-child phylogenetic networks.

\section{An approximation algorithm for the SPS}\label{sec4}
The only polynomial-time approximation algorithm known for the SPS in the literature is the following~\citep{fischer15}: let $\Gamma =\{1,\dots,n\}$ be a set of taxa, let $N$ be a rooted phylogenetic network of $\Gamma$ and let $C$ be a character of $\Gamma$. Let $s$ be the most frequently occurring state in $C$. Clearly, $s$ occurs on at least a $1/p$ fraction of the taxa. Let $T\in\mathcal{T}(N)$ and let $C'$ be an extension of $C$ to $V(N)$ defined by
\begin{align*}
C'(x)=\begin{cases}C(x)&\text{if $x\in\Gamma$},\\
s&\text{otherwise.}
\end{cases}
\end{align*}
Let OPT denote the minimum softwired parsimony score. Clearly, OPT$\,\geq p-1$. Hence, for $p\geq 2$,
\begin{align*}
\frac{\text{score}(T,C')}{\text{OPT}}\leq\frac{(1-1/p)n}{p-1}=\frac{n}{p}\leq\frac{n}{2},
\end{align*}
i.e., we obtain a linear approximation for the SPS. We call this approximation for the SPS the \emph{simple approximation} for the SPS. Given the very strong inapproximability results from~\citep{fischer15} this is (up to constant factors) best-possible, \CorrB{unless P $=$ NP}, on many classes of networks. However, in this section we will expand on our analysis from the last section to give a polynomial-time constant-factor approximation for the SPS restricted to phylogenetic networks from the class $\mathcal{N}$. To this end, we expand our analysis of rooted, binary phylogenetic trees from the last section to rooted, tree-child, binary phylogenetic networks. Hence, we consider the complete IP formulation for the SPS introduced by~\citet{fischer15}: due to the presence of reticulation vertices consider binary decision variables $y_e$ for all $e\in E_{\text{int}}^r$, such that
\begin{align*}
y_e=\begin{cases}1&\text{if edge $e$ appears in a feasible solution to the SPS},\\
0&\text{otherwise},
\end{cases}
\end{align*}
to obtain the following IP:
\begin{form}\label{form::sps}
\begin{align*}
\min~~\sum_{e\in E(T)}c_e~~~~~~~~~&\\
\text{s.t.}\,~~~~~~c_e-x_u^s+x_v^s&\geq 0~&~&\forall\,e=(u,v)\in E_{\text{int}}^t\cup E_{\text{ext}},s\in S_p\\
c_e+x_u^s-x_v^s&\geq 0~&~&\forall\,e=(u,v)\in E_{\text{int}}^t\cup E_{\text{ext}},s\in S_p\\
c_e-y_e-x_u^s+x_v^s&\geq -1~&~&\forall\,e=(u,v)\in E_{\text{int}}^r,s\in S_p\\
c_e-y_e+x_u^s-x_v^s&\geq -1~&~&\forall\,e=(u,v)\in E_{\text{int}}^r,s\in S_p\\
\sum_{s=0}^px_v^s&= 1~&~&\forall\,v\in V(T)\\
x_v^{C(v)}&=1~&~&\forall\,v\in V_{\text{ext}}\\
\sum_{(u,v)\in E_{\text{int}}^r}y_{(u,v)}&= 1~&~&\forall\,v\in V_{\text{int}}^r\\
c_e&\in\{0,1\}~&~&\forall\,e\in E(T)\\
y_e&\in\{0,1\}~&~&\forall\,e\in E_{\text{int}}^r\\
x_v^s&\in\{0,1\}~&~&\forall\,v\in V(T),s\in S_p
\end{align*}
\end{form}
Notice that Formulations~\ref{form::btps} and~\ref{form::sps} differ only in constraints encoding the presence/absence of reticulation edges. Analogously to the BTPS, we look at the dual of the LP relaxation of Formulation~\ref{form::sps}:
\begin{form}\label{form::spsDual}
\begin{align}
\max~~\sum_{v\in V(T)}\mu_v+\sum_{v\in V_{\text{ext}}}\eta_v+\sum_{(u,v)\in E_{\text{int}}^r}\left[\delta_v-\sum_{s=0}^p\left(\beta_{u,v}^s+\beta_{v,u}^s\right)\right]&\\
\text{s.t.}~~~~~~~~~~~~~~~~~~~~~~~~~~~~~~~~~~~~~~~~~~~~~~~~~~~~~~~~~~~~~~~~~~\sum_{s=0}^p\left(\beta_{u,v}^s+\beta_{v,u}^s\right)&\leq 1~&~&\forall\,(u,v)\in E_{\text{int}}^t\cup E_{\text{ext}}\\
\delta_v-\sum_{s=0}^p\left(\beta_{u,v}^s+\beta_{v,u}^s\right)&\leq 0~&~&\forall\,(u,v)\in E_{\text{int}}^r\\
\mu_v+\sum_{(v,w)\in E(T)}\left(\beta_{w,v}^s-\beta_{v,w}^s\right)+\sum_{(u,v)\in E(T)}\left(\beta_{u,v}^s-\beta_{v,u}^s\right)&\leq 0~&~&\forall\,v\in V_{\text{int}}^t\cup V_{\text{int}}^r,s\in S_p\label{con2::mu1}\\
\mu_v+\sum_{(u,v)\in E(T)}\left(\beta_{u,v}^s-\beta_{v,u}^s\right)&\leq 0~&~&\forall\,v\in V_{\text{ext}},s\in S_p\setminus\{C(v)\}\label{con2::mu2}\\
\mu_v+\eta_v+\sum_{(u,v)\in E(T)}\left(\beta_{u,v}^s-\beta_{v,u}^s\right)&\leq 0~&~&\forall\,v\in V_{\text{ext}},s=C(v)\label{con2::mu+eta}\\
\beta_{u,v}^s,\beta_{v,u}^s&\geq 0~&~&\forall\,(u,v)\in E(T),s\in S_p\\
\mu_v&\in\mathbb{R}~&~&\forall\,v\in V(T)\\
\eta_v&\in\mathbb{R}~&~&\forall\,v\in V_{\text{ext}}\\
\delta_v&\in\mathbb{R}~&~&\forall\,v\in V_{\text{int}}^r
\end{align}
\end{form}
We can simplify Formulation~\ref{form::spsDual} by exploiting the structure of rooted, binary phylogenetic networks like we did for trees and Formulation~\ref{form::btpsDual}. This means, for $v\in V_{\text{int}}^t$, $s\in S_p$, constraint~\eqref{con2::mu1} is of form~\eqref{con::bbb} and, for $v\in V_{\text{ext}}$, $s\in S_p$, constraints~\eqref{con2::mu2} and~\eqref{con2::mu+eta} are of form~\eqref{con::00b} and~\eqref{con::e00b}, respectively. Moreover, for $v\in V_{\text{int}}^r$, $s\in S_p$ with $(u_1,v),(u_2,v),(v,w)\in E(T)$, constraint~\eqref{con2::mu1} is of the form
\begin{align*}
\mu_v+\beta_{w,v}^s-\beta_{v,w}^s+\beta_{u_1,v}^s-\beta_{v,u_1}^s+\beta_{u_2,v}^s-\beta_{v,u_2}^s\leq 0.
\end{align*}
Similarly to our analysis in the last section, we first introduce all components that we need to construct an approximation algorithm for the SPS. Again, we make frequent use of complementary slackness conditions. In addition to conditions~\eqref{cs::dv1} to~\eqref{cs::pv4}, the following complementary slackness conditions can be associated with LP relaxation of Formulation~\ref{form::sps} and Formulation~\ref{form::spsDual}:
\begin{align}
\left(c_e-y_e-x_u^s+x_v^s+1\right)\beta_{u,v}^s&=0~&~&\forall\,e=(u,v)\in E_{\text{int}}^r,s\in S_p\label{cs::dv5}\\
\left(c_e-y_e+x_u^s-x_v^s+1\right)\beta_{v,u}^s&=0~&~&\forall\,e=(u,v)\in E_{\text{int}}^r,s\in S_p\label{cs::dv6}\\
\left(\sum_{(u,v)\in E_{\text{int}}^r}y_{(u,v)}-1\right)\delta_v&=0~&~&\forall\,v\in V_{\text{int}}^r\label{cs::dv7}\\
\left[\delta_v-\sum_{s=0}^p\left(\beta_{u,v}^s+\beta_{v,u}^s\right)\right]c_e&=0~&~&\forall\,e\in E_{\text{int}}^r\label{cs::pv5}\\
\left(\mu_v+\beta_{w,v}^s-\beta_{v,w}^s+\beta_{u_1,v}^s-\beta_{v,u_1}^s+\beta_{u_2,v}^s-\beta_{v,u_2}^s\right)x_v^s&=0~&~&\forall\,v\in V_{\text{int}}^r,s\in S_p\label{cs::pv6}
\end{align}

Recall from Observation~\ref{obs::triangle} that without loss of generality all phylogenetic networks in $\mathcal{N}$ are triangle-free. Hence, the graph $N\in\mathcal{N}$ looks locally around a reticulation vertex like the graph on the left in Figure~\ref{ret1}. This clear local structure allows us to extend Algorithm~\ref{PD-HPS-trees} to obtain an 2-approximation algorithm for the SPS for some phylogenetic networks:

\begin{prop}\label{prop::2approx}
Let $N\in\mathcal{N}$ \CorrB{(i.e. $N$ is a rooted, binary, tree-child network)} and let $C$ be a character of $\Gamma$. Then, the SPS \CorrB{instance formed by $\Gamma$, $N$ and $C$} can be approximated with factor 2 \CorrB{in polynomial time}.
\end{prop}
\begin{proof}
We can construct a solution to Formulations~\ref{form::sps} and~\ref{form::spsDual} by setting primal and dual variables to the same values as Algorithm~\ref{PD-HPS-trees} does when solving the SPS on any fixed induced phylogenetic subtree of $N$. This solution can be extended to a feasible primal and infeasible dual solution $F$ to the SPS for $N$ and $C$. Then, we can continue to apply Algorithm~\ref{PD-HPS-trees} until it is no longer possible to find an induced phylogenetic subtree of $N$ without fully resolved primal and dual variables in the resulting feasible solution $F$ to the SPS for $N$ and $C$. We call the state of all fully resolved primal and dual variables in $F$ \emph{optimal} because from Proposition~\ref{prop::treeOpt} we know that score$(T,C')$ is minimum when $T$ is an induced phylogenetic subtree of $N$ and $C'$ is the corresponding extension of character $C$ returned by Algorithm~\ref{PD-HPS-trees}.

First, assume $N$ is time-consistent. Recall from Observation~\ref{obs::time} that there exists some fixed order of $V(N)$ such that a propagation of character state sets from the children of vertices $v\in V(N)$ to $v$ in the order of $V(N)$ is well-defined. Then, there exists at least one reticulation vertex $v_r\in V_{\text{int}}^r$ with parents $v_1$ and $v_2$ such that, for $(v_r,w_r)$, $(v_1,w_1)$, $(v_2,w_2)\in E_{\text{int}}^t$, primal and dual variables associated with all vertices and edges present in the subtrees rooted in $w_r,w_1$ and $w_2$ are optimal (see Figure~\ref{ret1}).
In this proof we always choose $x_v^s=1/|C'(v)|$ for $v\in V(T)$, $s\in S_p$, $c_{(u,v)}=|C'(u)\setminus C'(v)|$, $(u,v)\in E(N)$, $\beta_{v,u}^s=0$ for $(u,v)\in E(N)$, $s\in S_p$, and $\beta_{u,v}^s=0$ for $y_{(u,v)}=0$, $(u,v)\in E_{\text{int}}^r$, $s\in S_p$. Define $N'$ as the subnetwork we obtain from $N$ by taking the union of the induced phylogenetic subnetworks rooted in $v_1$ and $v_2$.

\begin{description}
\item[Case 1:] $C'(w_r)=C'(w_1)=C'(w_2)$. Set $C'(v_1)=C'(v_2)=C'(v_r)=C'(w_r)$. Then, we set feasible values for dual variables $\beta_{v_1,w_1}^s,\beta_{w_1,v_1}^s$, $\beta_{v_2,w_2}^s,\beta_{w_2,v_2}^s$, $\beta_{v_r,w_r}^s,\beta_{w_r,v_r}^s$, $s\in S_p$ which satisfy complementary slackness conditions (see Proposition~\ref{prop::dualstep}). In addition, set $y_{(v_1,v_r)}=\delta_{v_r}=1$, $y_{(v_2,v_r)}=0$, $\beta_{v_1,v_r}^s=1/|C'(w_r)|$, $s\in C'(w_r)$, and $\beta_{v_1,v_r}^s=0$, $s\in S_p\setminus C'(w_r)$. Then, complementary slackness conditions~\eqref{cs::dv5} to~\eqref{cs::pv5} hold for edges $(v_1,v_r),(v_2,v_r)$ and vertex $v_r$. Furthermore, for vertex $v_r$, conditions~\eqref{cs::pv6} are equivalent to
\begin{align*}
\mu_{v_r}+0-\frac{1}{|C'(w_r)|}+\frac{1}{|C'(w_r)|}-0+0-0&=0~&~&\forall\,s\in C'(v_r),\\
\mu_{v_r}+0-0+0-0+0-0&\leq 0~&~&\forall\,s\in S_p\setminus C'(v_r).
\end{align*}
Hence, set $\mu_{v_r}=0$. Additionally, set $\mu_{v_1}=-1/|C'(w_1)|-1/|C'(w_r)|$ and $\mu_{v_2}=-1/|C'(w_2)|$. Thus, from the fact that conditions~\eqref{cs::pv2} hold for roots $v_1$, $v_2$, $s\in S_p$, we conclude that score$(N',C')$ is minimum.
\item[Case 2:] $C'(w_r)=C'(w_1)\neq C'(w_2)$. Set $C'(v_1)=C'(v_r)=C'(w_r)$ and $C'(v_2)=C'(w_2)$. Then, we can set primal and dual variables like in Case~1 to arrive at the same conclusion.
\item[Case 3:] $C'(w_r)\neq C'(w_1)$.
\begin{description}
\item[Case 3.1:] $C'(w_r)\subset C'(w_1)$. Then, apply Case~1 or~2 and recursively remove character states of the children of $C'(v_1)$ and $C'(v_2)$ which are not in $C'(w_r)$ until no more character states can be removed or the recursion arrives at a singleton character state set.
\item[Case 3.2:] $C'(w_r)\nsubseteq C'(w_1)$ and $C'(w_1)\subset C'(w_r)$. Then, apply Case~1 or~2 for $C'(v_1)=C'(v_2)=C'(v_r)=C'(w_1)$ or $C'(v_1)=C'(v_r)=C'(w_1)$, $C'(v_2)=C'(w_2)$, respectively. Then, recursively remove character states of the children of $C'(v_1)$ which are not in $C'(w_1)$ until no more character states can be removed or the recursion arrives at a singleton character state set.
\item[Case 3.3:] $C'(w_r)\nsubseteq C'(w_1)$ and $C'(w_1)\nsubseteq C'(w_r)$. Then, apply Case~1 or~2 for $C'(v_1)=C'(v_2)=C'(v_r)=C'(w_1)\cup C'(w_r)$ or $C'(v_1)=C'(v_r)=C'(w_1)\cup C'(w_r)$, $C'(v_2)=C'(w_2)\cup C'(w_r)$, respectively, except for the definition of $\mu_{v_1}$ and $\mu_{v_2}$. This means, $C'(w_1)$ and $C'(w_2)$ are proper subsets of $C'(v_1)$ and $C'(v_2)$, respectively. Since $C'(w_1)\subset C'(v_1)=C'(v_r)$, i.e., for $(u,v_1)\in E(N)$, conditions~\eqref{cs::pv2} include equations
\begin{align*}
\mu_{v_1}+\beta_{w_1,v_1}^s-0+\beta_{v_r,v_1}^s-0+0-\beta_{v_1,u}^s&=0~&~&\forall\,s\in C'(w_1),\\
\mu_{v_1}+0-0+\beta_{v_r,v_1}^s-0+0-\beta_{v_1,u}^s&=0~&~&\forall\,s\in C'(v_1)\setminus C'(w_1).
\end{align*}
Hence, there exists no $\mu_{v_1}$ to complete the complementary slackness conditions for our previously defined dual variables. An analogous argument holds for $\mu_{v_2}$.
\end{description}
\end{description}

If Case~3.3 did not occur, then we know from Proposition~\ref{prop::dualstep} that all primal and dual variables we have defined so far satisfy their respective complementary slackness conditions. Otherwise, we choose $\mu_{v_1}$ and $\mu_{v_2}$ maximum, i.e., our constructed solution can be extended to a feasible solution for the SPS. We call the state of all fully resolved primal and dual variables so far optimal again even if Case~3.3 prevents some complementary slackness conditions to be fulfilled.

After we have finished processing the reticulation vertex $v_r$ as described above, either there exists another reticulation vertex that satisfies the same assumptions as our choice of $v_r$ did or there exist vertices $v_1,v_2$ and $v_r$ like in Figure~\ref{ret1} and directed paths $P_1,P_2$ and $P_r$ starting in $v_1,v_2$ and $v_r$ and ending in vertices $v_1',v_2'$ and $v_r'$, whose associated primal and dual variables are optimal in $N[v_1'],N[v_2'] $ and $N[v_r']$, respectively. In the latter case, we make our new choice for a reticulation vertex $v_r$ such that all vertices in $P_1,P_2$ and $P_r$ are adjacent to vertices $p$ whose associated primal and dual variables are optimal in $N[p]$. This choice is always possible because $N$ is time-consistent. Hence, we can apply dual and primal steps like in Algorithm~\ref{PD-HPS-trees} to vertices in $P_1,P_2$ and $P_r$ non-increasing a time-stamp function of $N$ to establish local optimality for primal and dual variables on these paths. Thus, we can repeat the procedure for reticulation vertices we have outlined so far to process all vertices of $N$ in a well-defined propagation sequence of character states from the leaves to character state sets on internal vertices. Finally, we finish our computations by fixing $C'(\rho)$ to be a singleton and applying another primal step. 

If Case~3.3 never occurred, then all complementary slackness conditions of Formulations~\ref{form::sps} and~\ref{form::spsDual} are satisfied for our constructed feasible solution. Thus, our construction yields an optimal solution to the SPS. Otherwise, we know from Proposition~\ref{prop::dualstep}.4.i that our definition of primal and dual variable still do not satisfy complementary slackness. In the worst case $C'(w_1)\neq C'(v_1)$ and $C'(w_2)\neq C'(v_2)$ for each set of parents $v_1,v_2$ of a reticulation vertex $v_r$ for which we applied Case~3.3 because this leads to the maximum contribution of $2$ to the objective function value of our solution. Without loss of generality the best case occurs for $C'(w_1)\neq C'(v_1)$ and $C'(w_2)=C'(v_2)$ because $C'(w_r)\neq C'(w_1)$. This means, the contribution of vertices $v_1,v_2$ and $v_r$ to the objective function value is at most twice this contribution in an optimal solution to the SPS.

\begin{figure*}[pos=!b,align=\centering]
\centering
\includegraphics[scale=0.4]{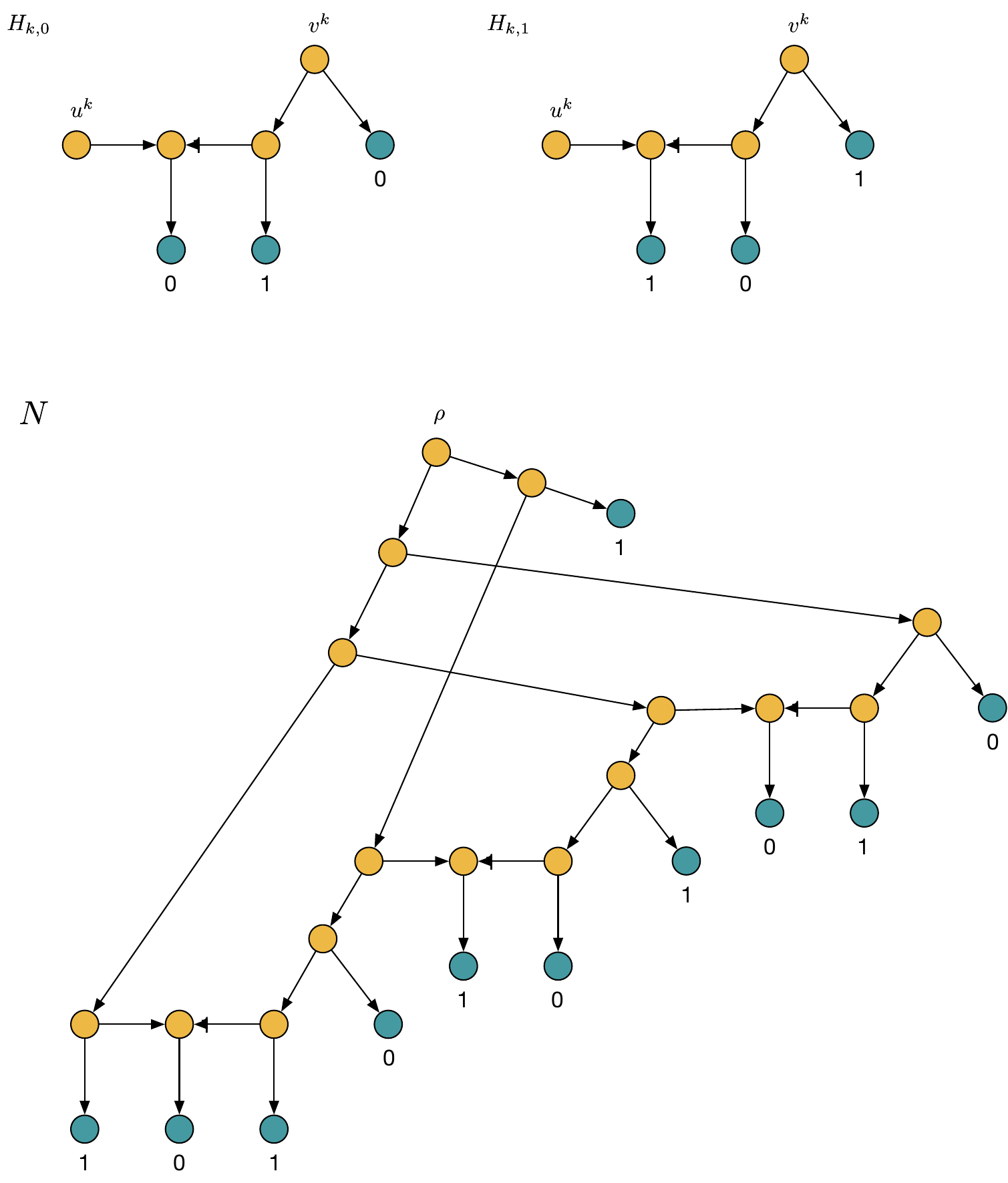}
\caption{Two directed graphs $H_{k,0}$ and $H_{k,1}$ with partially labelled vertices and, for $k=3$, a phylogenetic network $N$ on $3k+2$ taxa with root $\rho$ such that $H_{1,0},H_{2,1}$ and $H_{3,0}$ are induced subgraphs of $N$ when suppressing labels of internal vertices.}\label{fig::ta1}
\end{figure*}

Now, if $N$ is not time-consistent, then we can apply the same propagation sequence of character state sets to vertices $v\in V(N)$ as in the case that $N$ is time-consistent until we encounter a reticulation vertex $v_r$ with $(v,v_r)\in E_{\text{int}}^r$ for which our arguments do apply anymore. Then, there exists a vertex $v'\in V(N)$ with a directed path from $v$ to $v'$ in $N$ such that $v'$ has a child which is reticulation vertex $v_r'$ and $v_r'$ is unprocessed by our propagation algorithm so far. Such vertices $v,v'$ exist because $N$ is tree-child. If we can find this vertex $v'$ for both parents $v$ of $v_r$, then we reapply our arguments for one choice of $v'$ instead of $v$ until exactly one parent of a reticulation vertex fits into the definition of vertex $v$. Again, reapplying our arguments to identify suitable vertices $v$ and $v'$ is possible because $N$ is tree-child. Then, we choose $v_2=v$ in our procedure above to process $v_r$ while keeping $C'(v_2)$ undefined. This allows us to process a reticulation vertex preventing $N$ from being time-consistent by keeping one of it's parents unprocessed without incurring any additional cost different from the additional cost of Case~3.3. Thus, we reach the same conclusion as for a time-consistent network $N\in\mathcal{N}$.
\end{proof}

Observe that the quality of approximation in Proposition~\ref{prop::2approx} solely relies on the occurence of Case~3.3. Hence, we can improve the approximation by avoiding Case~3.3 whenever possible (for example by swapping the role of $w_1$ and $w_2$ in the proof of Proposition~\ref{prop::2approx}). Furthermore, Proposition~\ref{prop::dualstep}.3 indicates that alternative constructions to the one chosen in Case~3.3 lead to the same problems with local non-optimality. Moreover, the approximation factor of $2$ is tight for the class $\mathcal{N}$:

\begin{prop}
There exists an integer $n\geq 5$, a network $N\in\mathcal{N}$ and a character $C$ of $\Gamma$ such that the algorithm detailed in Proposition~\ref{prop::2approx} yields $T\in\mathcal{T}(N)$ and an extension $C'$ of $C$ from $\Gamma$ to $V(N)$ with score$(T,C')=2\,\text{OPT}$.
\end{prop}
\begin{proof}
First, for $k\geq 1$, we introduce graphs $H_{k,0}$ and $H_{k,1}$ with vertices labelled as in Figure~\ref{fig::ta1}. Then, for fixed integers $n=3k+2$, $k\geq 1$, we construct a phylogenetic network $N$ as follows: consider graphs $H_{1,i_1},\dots,H_{k,i_k}$ such that $i_1=0,i_2\dots,i_k\in\{0,1\}$ take alternating values from $i_1$ to $i_k$. Then, add vertex $v^0$ with label $i_0=1$ and, for $l=1,\dots,k$, add the edge $(u^l,v^{l-1})$. Next, add a vertex $u^{k+1}$ with label $i_{k-1}$. Furthermore, if $k$ is odd, then add rooted binary trees $T_1$ and $T_2$ with leafset $\{u^1,u^3,\dots,u^k\}$ and $\{u^2,u^4,\dots,u^{k+1}\}$, respectively. If $k$ is even, then add rooted binary trees $T_2$ and $T_1$ with leafset $\{u^2,u^4,\dots,u^k\}$ and $\{u^1,u^3,\dots,u^{k+1}\}$, respectively. Finally, add the root vertex $\rho$ whose children are the roots of $T_1$ and $T_2$. An example of this construction of $N$ for $k=3$ is shown in Figure~\ref{fig::ta1}. 

We define $C$ as the character taking states equal to the labels of leaves in $N$. Then, all reticulation vertices of $N$ appear in subgraphs $H_{l,i_l}$, $l\in\{1,\dots,k\}$ and in the algorithm detailed in Proposition~\ref{prop::2approx} they all fit into Case~3.3. Hence, all leaves of the induced subtrees $T_1$ and $T_2$ have character state $1$ and $0$ or $0$ and $1$, respectively. Thus, the children of $\rho$ have character states $0$ and $1$, respectively, yielding score$(T^{\text{approx}},C^{\text{approx}})=2k+1$. Analogously, we get $\text{OPT}=k+2$. To see this, note that, for $l=1,\dots,k$, removing vertices $u^l$ as parents of the reticulation vertices of $N$ gives us $T^{\text{OPT}}\in\mathcal{N}$ with score$(T^{\text{OPT}},C^{\text{OPT}})=k$ in the subnetwork $N'$ which is induced by $N$ by removing the root $\rho$ and the non-leaf vertices of $T_1$ and $T_2$. Hence, all leaves of $T_1$ and $T_2$ have character state $1$ and $0$ or $0$ and $1$, respectively, except for either one leaf in $T_1$ or $T_2$ whose character state differs from all other leaves in $T_1$ or $T_2$. respectively.

In total, we conclude that there exists $N\in\mathcal{N}$ and a character $C$ of $\Gamma$ such that
\begin{align*}
\frac{\text{score}(T,C')}{\text{OPT}}=\frac{2k+1}{k+2}=\frac{2(n-2)/3}{(n-2)/3+2}=2-\frac{12}{n+4}\stackrel{n\to\infty}{\longrightarrow}2.
\end{align*}
\end{proof}

\section{Concluding remarks}
\label{sec:5}
In this article we have investigated the softwired parsimony problem for rooted, binary, tree-child phylogenetic networks and any character to define a polynomial time 2-approximation algorithm for this problem; \CorrB{we have also shown that this approximation guarantee is tight}. This novel approximation algorithm vastly improves the linear approximation quality guarantee of the simple approximation for the SPS which is known in the literature as the best approximation algorithm for the SPS in general.
Plausibly our 2-approximation algorithm could be used as part of a combinatorial branch and bound approach to solving the SPS exactly on tree-child networks.

To the best of our knowledge, this is with the very recent exception of \cite{deen2023near} the only known example of an explicit primal-dual algorithm in phylogenetics. This is interesting, since the polyhedral angle on phylogenetics problems dates back to at least 1992 \cite{erdos1992evolutionary} and integer linear programming has been used quite frequently in the field (see e.g. \cite{sridhar2008mixed}). We suspect that unpacking many of these integer linear programming formulations to study properties of the linear programming relaxations could be a very fruitful line of research. The recent duality-based 2-approximation algorithm for \emph{agreement forests} \cite{olver2023duality} is a good example of this potential.

\CorrB{Three explicit questions to conclude. First, in how far can the result in this paper be generalized beyond the subclass of rooted, binary, tree-child networks? In the proof of Proposition~\ref{prop::2approx} we have seen that the particular local structure of tree-child networks allows us to lift the restriction of our approximation algorithm to time-consistent networks. The same argument does not extend to superclasses of tree-child networks. However, since we derive our approximation algorithm from an algorithmic construction on trees (Proposition~\ref{prop::fitch2}) subclasses of tree-based networks are suitable candidates to generalize our results~\cite{kong1}. In particular, studying orchard networks in this context might be fruitful as the \emph{horizontal gene transfer labelling} characterizing them~\cite{leo22} does remove possible complications in the analysis when dealing with time-consistency beyond tree-child networks.} Second, in how far can it be generalized to other notions of network parsimony \cite{van2018improved}, or even to non-parsimony models for scoring networks? \CorrB{Third, can we find a polynomial time approximation algorithm for SPS on rooted, binary, tree-child networks that achieves an approximation factor better than 2? At present we do not know whether the problem is APX-hard, so it could still even be the case that
the problem permits a polynomial-time approximation scheme (PTAS).}

\section*{Acknowledgments}
\noindent \small{The first author acknowledges support from the European Union’s Horizon 2020 research and innovation programme under the Marie Skłodowska-Curie grant agreement no. 101034253, and by the NWO Gravitation project NETWORKS under grant no. 024.002.003.}

\section*{Compliance with Ethical Standards}
\noindent \small{Funding: This study was funded by the European Union’s Horizon 2020 research and innovation programme under the Marie Skłodowska-Curie grant agreement (101034253) and the NWO Gravitation project NETWORKS (024.002.003)}

\noindent \small{Ethical approval: This article does not contain any studies with human participants or animals performed by any of the authors.}

\bibliographystyle{cas-model2-names}

\end{document}